\documentclass[11pt]{article}

\include{_preamble}

\makeatletter
\@ifclassloaded{svjour3}{%
  
\begin{document}
  \title{\titlepaper
  }
  \titlerunning{Evaluating Longitudinal Modified Treatment Policies Under
    Competing Risks}
  \date{Received: date / Accepted: date}
  \maketitle
}{%
  \title{\titlepaper}
  \date{\today}
  \author{\authorlist}
  \begin{document}
  \maketitle
}
\makeatother

\begin{abstract}
  
  Longitudinal modified treatment policies (LMTP) have been recently
  developed as a novel method to define and estimate causal parameters
  that depend on the natural value of treatment. LMTPs represent an
  important advancement in causal inference for longitudinal studies
  as they allow the non-parametric definition and estimation of the
  joint effect of multiple categorical, ordinal, or continuous
  treatments measured at several time points. We extend the LMTP
  methodology to problems in which the outcome is a time-to-event
  variable subject to a competing event that precludes observation of
  the event of interest. We present identification results and
  non-parametric locally efficient estimators that use flexible
  data-adaptive regression techniques to alleviate model
  misspecification bias, while retaining important asymptotic
  properties such as $\sqrt{n}$-consistency. We present an application
  to the estimation of the effect of the time-to-intubation on acute
  kidney injury amongst COVID-19 hospitalized patients, where death by
  other causes is taken to be the competing event.
\end{abstract}

\section{Introduction}

In survival analysis, a competing event is one whose realization
precludes the occurrence of the event of interest. As an example,
consider a study on the effect of mechanical ventilation of
hospitalized COVID-19 patients on the onset of acute kidney injury,
where the death of a patient by extraneous causes may be considered a
competing event. Commonly used methods for analyzing time-to-event
outcomes under competing events include estimating the cause-specific
hazard function~\citep{prentice1978analysis}, the sub-distribution
function~\citep[also known as cumulative
incidence;][]{fine1999proportional}, or treating the competing event
as a censoring event. Early techniques for the estimation of these
parameters were developed in the context of the Cox proportional
hazards models~\citep{cox1972regression, andersen1982cox}. More recent
methods for analyzing survival data include joint
modeling~\citep{henderson2000joint} and a multitude of other
parametric and semi-parametric
methods~\citep[e.g.,][]{zeger1992overview}.

While these methods have allowed much progress in applied research,
they have two important limitations: (i) they lack a formal framework
that clarifies the conditions required for a causal interpretation of
the estimates, and (ii) they typically rely upon parametric modeling
assumptions that can induce non-negligible bias in the effect
estimates, even when the assumptions required for causal
identification hold. Recent developments in causal inference and
semi-parametric estimation have yielded important progress in
addressing these significant limitations. For example, when the
treatment is binary,~\citet{young2020causal} present a study of the
definition and identification of several causal effect parameters in
the presence of competing risks, and~\citet{benkeser2018improved}
and~\citet{rytgaard2021one} describe semi-parametric--efficient
estimators that leverage the flexibility of data-adaptive regression
procedures in order to mitigate model misspecification bias while
retaining asymptotic properties critical to reliable statistical
inference (e.g., $\sqrt{n}$-consistency). Among other
insights,~\cite{young2020causal} clarify that cumulative incidence
effects may be interpreted as a total effect of exposure operating
through pathways that include the competing event, whereas treating
the competing event as a censoring event quantifies direct effects
operating independent of the competing event. Furthermore, they
clarify that identification when treating the competing event as a
censoring event entails considering interventions that eliminate the
competing event, and therefore require no unmeasured confounding of
the competing event--outcome relation. This assumption is not required
to identify the cumulative incidence effect, which is the focus of
this work. The reader interested in further issues surrounding the
definition and usefulness of these effects in different applications
is encouraged to read the original research article
\citep{young2020causal}.

Under a counterfactual causal framework, defining causal effects
requires considering outcomes under hypothetical interventions on the
treatment of interest. In the case of binary treatments, causal
effects are often defined as differences between the expected outcomes
under hypothetical interventions that assign the treatment to all
units or to none --- that is, the average treatment effect. Causal
effects defined through such deterministic interventions are limited
and or lack a basis for scientific interpretation in many
applications, including problems where the treatment of interest is
numerical rather than categorical, where the treatment is
multivariate, where the treatment is a time-to-event-variable (such as
time-to-intubation as in our illustrative application), or when
treatment assignment is deterministic conditional on potential
confounders (cf.~positivity/overlap assumptions).

As a solution to the definition of relevant causal effects for these
problems, recent work has focused on stochastic interventions and
modified treatment policies~\citep{stock1989nonparametric,
  robins2004effects, Diaz12, Haneuse2013,richardson2013single,
  young2014identification}. Stochastic interventions inquire what
would have happened in a hypothetical world where the
post-intervention treatment is a random draw from a user-specified
distribution conditional on covariates. Examples of stochastic
interventions include interventional mediation
effects~\citep{vanderweele2014effect, diaz2021nonparametric,
  hejazi2022nonparametric}, as well as incremental propensity score
interventions whose identification does not require the positivity
assumption~\citep{kennedy2019nonparametric,
  wen2021intervention}. Modified treatment policies (MTPs) generalize
stochastic interventions and allow the post-intervention treatment to
also depend on the natural value of treatment~\citep[for an overview
of such interventions, see][]{young2014identification}. MTPs allow the
definition and estimation of parameters defined in terms of natural
research questions, such as inquiring what the incidence of wheezing
among children with asthma would have been had the concentration of
PM\textsubscript{2.5} particles been 5\% lower than they actually
were, or what the mortality rate due to opioid overdose would have
been had naloxone access laws been implemented a year
earlier than they were actually implemented~\citep{rudolph2021effects}.

Identifying assumptions for causal effects of longitudinal modified
treatment policies using the extended g-formula were first articulated
by~\citet{richardson2013single}
and~\citet{young2014identification}. Recently,~\citet{diaz2021lmtp}
developed sequential regression and sequentially doubly robust
estimators for longitudinal modified treatment policies (LMTP). These
estimators are generalizations of estimators developed
by~\citet{diaz2018stochastic} for MTPs in the single time-point case,
and of estimators developed by~\citet{luedtke2017sequential}
and~\citet{rotnitzky2017multiply} for static interventions in the
multiple time-point case. Here, we extend the LMTP methodology for the
estimation of the cumulative incidence under competing
events, illustrating how LMTPs can be used to solve problems whose
solution has remained elusive in causal inference, such as
non-parametric definition and estimation of effects for continuous,
multi-valued, time-valued, or multivariate exposures in the context
of time-varying exposures. In our motivating example we illustrate
estimation of the effect of time-to-intubation on time to acute kidney
injury with death by other causes acting as a competing risk.

Our estimation strategy is rooted in recent developments allowing the estimation
of causal effects using flexible regression techniques from machine
learning~\citep{vanderLaanRose11, vanderLaanRose18, chernozhukov2017double}.
Central to this theory is the notion of von-Mises
expansion~\citep[e.g.,][]{mises1947asymptotic, vanderVaart98,
robins2009quadratic}, which together with sample-splitting
techniques~\citep{bickel1982adaptive, klaassen1987consistent, zheng2011cross}
allows the use of flexible regression methods for the estimation of nuisance
parameters while still allowing the construction of efficient and
$\sqrt{n}$-consistent estimators.

\section{Defining LMTP effects under competing risks}\label{sec:notation}

Assume that at each time point $t=1,\ldots\tau$, we measure
$X_t=(D_t, Y_t, L_t, A_t, C_t)$ on each of $i=1,\ldots, n$ study
participants, where $L_t$ denotes time-varying covariates, $A_t$
denotes a vector of treatments at time $t$, $C_t$ denotes an indicator
of loss-to-follow-up (equal to one if the unit remains on study at
time $t+1$ and equal to zero otherwise), $D_t$ denotes an indicator of
the competing event occurring on or before time $t$, and \corr{$Y_t$
  denotes an indicator of the event of interest not having occurred on
  or by time $t$.} Baseline covariates are included in $L_1$. At the
end of the study, we measure the outcome of interest $Y_{\tau+1}$. We
let \corr{$D_1=0$ and $Y_1=1$,} meaning that participants have not
experienced the event of interest nor the competing event at the
beginning of the study. We let \corr{$R_t=\one\{D_t=0, Y_t=1\}$}
denote an indicator that neither the event of interest nor the
competing event have occurred by time $t$, \corr{and
  $Z_t=\one\{D_t=1, Y_t=1\}$ denote an indicator that the event of
  interest has not occurred but the competing event has}. Assume
monotone loss-to-follow-up so that $C_t=0$ implies $C_k=0$ for all
$k>t$, in which case all data become degenerate for $k > t$. In this
work, we are interested in a scenario where occurrence of the
competing event precludes occurrence of the event of interest. For
example, in a study looking to assess the effect of hypertension
medications on the likelihood of a stroke, death for reasons other
than a stroke is a competing event. In this notation, we have
\corr{$Y_t=0$ if $Y_{t-1}=0$}, $D_t=1$ if $D_{t-1}=1$, \corr{and $Y_t=1$ if
$D_t=1$ and $Y_{t-1}=1$}, since the competing event precludes occurrence of the event
of interest. Let $X_i = (X_{1,i}, \ldots, X_{\tau, i})$ denote the
observed data for subject $i$ across all time points
$t=1,\ldots\tau$. We let $\P f = \int f(x)\dd \P(x)$ for a given
function $f(x)$. We use $\Pn$ to denote the empirical distribution of
$X_1,\ldots\,X_n$, and assume $\P$ is an element of the nonparametric
statistical model defined as all continuous densities on $X$ with
respect to a dominating measure $\nu$. We let $\E$ denote the
expectation with respect to $\P$, i.e.,
$\E\{f(X)\} = \int f(x)\dd\P(x)$. We also let $||f||^2$ denote the
$L_2(\P)$ norm $\int f^2(x)\dd\P(x)$. We use
$\bar W_t = (W_1,\ldots, W_t)$ to denote the past history of a
variable $W$, use $\ubar W_t = (W_t,\ldots, W_\tau)$ to denote the
future of a variable, and use
$H_t = (\bar D_t, \bar Y_t, \bar L_t, \bar A_{t-1})$ to denote the
history of the time-varying covariates and the treatment process until
just before time $t$. $\bar W$ denotes the entire history
$(W_1,\ldots,W_\tau)$ of a variable. By convention, variables with an
index $t\leq 0$ are defined as the null set, expectations conditioning
on a null set are marginal, products of the type
$\prod_{t=k}^{k-1}b_t$ and $\prod_{t=0}^0b_t$ are equal to one, and
sums of the type $\sum_{t=k}^{k-1}b_t$ and $\sum_{t=0}^0b_t$ are equal
to zero. For vectors $u$ and $v$, $u\le v$ denotes point-wise
inequalities.

We formalize the definition of the causal effects using a
non-parametric structural equation model
\citep{Pearl00}. Specifically, for each time point $t$, we assume the
existence of deterministic functions $f_{D,t}$, $f_{Y,t}$, $f_{L,t}$,
$f_{A,t}$, $f_{C,t}$, such that
\begin{align*}
  D_t&=f_{D,t}(C_{t-1}, A_{t-1}, H_{t-1}, U_{D,t})\\
  Y_t&=f_{Y,t}(D_t, C_{t-1}, A_{t-1}, H_{t-1}, U_{Y,t})\\
  L_t&=f_{L,t}(Y_t, D_t, C_{t-1}, A_{t-1}, H_{t-1}, U_{Y,t})\\
  A_t&=f_{A,t}(H_t, U_{A,t})\\
  C_t&=f_{C,t}(A_t, H_t, U_{C,t}).
\end{align*}
Here $U=(U_{L,t}, U_{A,t}, U_{C,t}, U_{D,t},U_{Y,t}:t\in
\{1,\ldots,\tau+1\})$ is a vector of exogenous variables. Independence
assumptions on these errors necessary for identification will be
clarified in \S\ref{sec:iden}. Without loss of generality, we assume
that variables which are undefined are equal to zero (e.g., $L_t=0$ if
$C_k=1$ for any $k<t$).



We are concerned with the definition and estimation of the causal
effect of an intervention on the treatment process $\bar A$ on the
event indicator $Y_{\tau+1}$ in a hypothetical world where there is no
loss to follow-up (i.e., $\bar C=0$). The interventions on the
treatment process are defined in terms of longitudinal modified
treatment policies \citep{Diaz12,Haneuse2013}, which are hypothetical
interventions where the treatment is assigned as a new random variable
$A_t^\d$ (which may depend on the natural value of treatment as
explained below) instead of being assigned according to the structural
equation model. An intervention that sets the treatments up to time
$t-1$ to $\bar A_{t-1}^\d$ generates a counterfactual variable
$A_t(\bar A_{t-1}^\d)$, which is referred to as the \textit{natural
  value of treatment} \citep{richardson2013single,
  young2014identification}, and represents the value of treatment that
would have been observed at time $t$ under an intervention carried out
up until time $t-1$ but discontinued thereafter. An intervention on
all the treatment and censoring variables up to $t=\tau$ generates a
counterfactual outcome $Y_{\tau+1}(\bar A^\d)$. Causal effects are
defined as contrasts between the counterfactual probability
$\P[Y_{\tau+1}(\bar A^\d)=1]$ \corr{of not experiencing a failure
  event by the end of the study} under interventions implied by
different interventions $\d$, as defined below.

The causal effects we define are characterized by a user-given
function $\d(a_t, h_t,\varepsilon_t)$ that maps a given treatment
value $a_t$ (i.e., the ``natural value'') and a history $h_t$ into a
new treatment value.  The function $\d$ is also allowed to depend on a
randomizer $\varepsilon_t$.  This definition generalizes static
interventions (if $\d(a_t,h_t,\varepsilon_t)$ is a constant), dynamic
interventions (if $\d(a_t,h_t,\varepsilon_t)$ depends on $h_t$ but not
on $a_t$ nor $\varepsilon_t$), and stochastic interventions (if
$\d(a_t,h_t,\varepsilon_t)$ depends on $h_t$ and $\varepsilon_t$ but
not on $a_t$). In what follows we will assume that the randomizer is
drawn independently across units and independently of $U$, and whose
distribution does not depend on $\P$. 
For fixed values $\bar a_t$, and $\bar l_t$, we
recursively define $a_t^\d=\d(a_t, h^\d_t, \varepsilon_t)$, where
$h^\d_t=(\bar a_{t-1}^\d, \bar l_t)$. The intervention is undefined if
$h^\d_t$ is such that $d_s^\d=1$ or $y_s^\d=1$ for any $s\leq t$. The
LMTPs that we study are thus defined as $A^\d_t = \d(A_t(\bar
A_{t-1}^\d), H_t(\bar A_{t-1}^\d), \varepsilon_t)$ for a user-given
function $\d$. In what follows, we let $\g_{A,t}(a_t \mid h_t)$ denote
the density or probability mass function of $A_t$ conditional on
$C_{t-1}=R_t=1$ and $H_t=h_t$, and we let $\g_{C,t}(a_t, h_t)$ denote
$\P(C_t=1\mid C_{t-1}=R_t=1, A_t=a_t, H_t=h_t)$. Furthermore, we use
$\g_t^\d(a_t \mid h_t)$ to denote the density or probability mass
function of $\d(A_t,H_t,\varepsilon_t)$ conditional on $C_{t-1}=R_t=1$
and $H_t=h_t$.

To ground the ideas and illustrate the usefulness of this framework, we now
consider several examples of regimes $\d$ that yield interesting and
scientifically informative causal contrasts. More examples appear in
\citet{robins2004effects,young2014identification,diaz2021lmtp}.

\begin{example}[Additive shift LMTP]
  Let $A_t$ denote a vector of numerical treatments, such as a drug dose or
  energy expenditure through physical activity. Assume that $A_t$ is supported
  as $\P(A_t \leq u_t(h_t)\mid H_t = h_t)=1$ for some $u_t$. Then, for
  a user-given value $\delta$, we let
  \begin{equation}\label{eq:defdshift}
    \d(a_t,h_t)=
    \begin{cases}
      a_t + \delta & \text{if } a_t+\delta \leq u_t(h_t) \\
      a_t & \text{if } a_t +\delta > u_t(h_t).
    \end{cases}
  \end{equation} We define
  $A_t^\d=\d(A_t(\bar A_{t-1}^\d), H_t(\bar A_{t-1}^\d))$. This
  intervention has been discussed by~\citet{Diaz12},
  \citet{diaz2018stochastic}, and~\citet{Haneuse2013} in the context
  of single time-point treatments; by~\citet{diaz2020causal}
  and~\citet{hejazi2022nonparametric} in the context of causal
  mediation analysis; and by~\citet{hejazi2020efficient} in the
  context of studies with two-phase sampling designs. This
  intervention considers hypothetical worlds in which the natural
  treatment at time $t$ is increased by a user-given value $\delta$,
  whenever such an increase is feasible for a unit with history
  $H_t(A_{t-1}^\d)$. Specifically, note that if the treatment is
  supported in the real numbers, the MTP could simply be defined as
  $\d(a_t,h_t) = a_t+\delta$. However, when the treatment is supported
  as $\P(A_t \leq u_t(h_t)\mid H_t = h_t)=1$, the intervention
  $\d(a_t,h_t) = a_t+\delta$ will not be feasible for units for which
  $a_t+\delta > u_t(h_t)$. We therefore set $\d(a_t,h_t) = a_t$ for
  those units. Note that setting $\d(a_t,h_t) = u_t(h_t)$ could also
  be an interesting option for these units. However, this option would
  lead to a parameter that violates Condition~\ref{ass:inv} below,
  which is necessary to achieve $\sqrt{n}-$consistent estimation.
\end{example}

\begin{example}[Multiplicative shift LMTP]\label{ex:mshift}
  Let $A_t$ denote a vector of numerical treatments such as pollutant
  concentrations for various pollutants. To define this intervention, assume
  that $A_t$ is supported as $\P(A_t \geq l_t(h_t)\mid H_t = h_t)=1$ for some
  $l_t$. Then, for a user-given value $0 < \delta < 1$, we let
  \begin{equation}\label{eq:mshift}
    \d(a_t,h_t)=
    \begin{cases}
      a_t \times \delta & \text{if } a_t \times \delta \geq l_t(h_t)  \\
      a_t & \text{if } a_t \times \delta < l_t(h_t).
    \end{cases}
  \end{equation}
  The intervention implied by this function considers hypothetical
  worlds where the natural treatment at time $t$ is increased by a
  user-given factor $\delta$, whenever such increase is feasible for a
  unit with history $H_t(A_{t-1}^\d)$. This intervention seems useful,
  for example, to solve current problems in environmental epidemiology
  related to the effect of environmental mixtures
  \citep{gibson2019overview}.
\end{example}

\begin{example}[Incremental propensity score interventions based on
  the odds and risk ratio]\label{ex:ipsi}
  \citet{kennedy2019nonparametric} proposed an intervention tailored
  to binary treatments in which the post-intervention treatment is a
  draw from a user-given distribution $\g_t^\d(a_t\mid h_t)$, where
  this distribution is chosen such that the odds ratio comparing the
  likelihood of treatment under intervention versus the actual
  treatment is a user-given value $\delta$. Specifically, consider
  $\varepsilon_t\sim U(0,1)$, and define $\d(h_t,
  \varepsilon_t)=\one\{\varepsilon_t< \g_t^\d(1\mid h_t)\}$ where we
  define $\g_t^\d(1\mid h_t) = \delta \g_t(1\mid h_t)/\{\delta
  \g_t(1\mid h_t)+\g_t(0\mid h_t)\}$. Identifying the effect of this
  intervention does not require the positivity assumption. Thus,
  incremental propensity score interventions are very useful in
  settings where violations of the positivity assumption preclude
  identification of static effects such as the average treatment
  effect.

  Recently, \cite{wen2021intervention} proposed a similar
  intervention that instead uses the risk ratio to quantify the
  likelihood of treatment under the intervention. Their intervention
  is defined in terms of a random draw from a distribution
  $\g_t^\d(a_t\mid h_t) = a_t\delta \g_t(1\mid h_t) +
  (1-a_t)(1-\delta \g_t(1\mid h_t))$ (compared to
  \cite{wen2021intervention} we have flipped the values 0 and 1 for
  ease of exposition). The following LMTP yields the same
  identifying functional as that of  \cite{wen2021intervention}:
  \begin{equation}\label{eq:defipsi}
    \d(a_t,\varepsilon_t)=
    \begin{cases}
      a_t & \text{if } \varepsilon_t < \delta  \\
      0 & \text{otherwise},
    \end{cases}
  \end{equation}
  where $\varepsilon_t$ is a random draw from the uniform distribution
  on $(0,1)$. As before, we make use of the definition
  $A_t^\d=\d(A_t(\bar A_{t-1}^\d), H_t(\bar A_{t-1}^\d))$. Note that
  the probability of treatment under the intervention is equal to
  $\g_\delta(1\mid h_t) = \delta \times \g(1\mid h_t)$; thus, giving a
  risk-ratio interpretation to $\delta$.  Importantly, a unit with
  zero probability of treatment conditional on their history also has
  zero probability of treatment under the intervention. This weakens
  the positivity assumptions required for identification and
  estimation, as we will discuss below.
\end{example}

\begin{example}[Dynamic treatment initiation strategies with grace
  period]\label{ex:grace}
  Let $A_t$ denote a binary treatment variable.  Consider an intervention of the
  form ``for a user-given grace period $m$, if a condition for treatment
  initiation is met at or before $t-m$ then start treatment at $t$; otherwise,
  do not intervene.'' Letting $L'_t$ denote an indicator that the condition has
  been met, this intervention can be operationalized as an LMTP as follows:
  \begin{equation}\label{eq:defgrace}
    \d(a_t,h_t)=
    \begin{cases}
      1 & \text{if } l'_{t-m} = 1  \\
      a_t & \text{otherwise}.
    \end{cases}
  \end{equation}
  This type of intervention has been considered in the context of treatment
  initiation for HIV~\citep{van2005history,cain2010start}, where the condition
  $L_t'$ is that CD4+ cell count drops below a pre-specified value. In this type
  of application, realistic treatment policies may involve a delay of treatment
  initiation for logistical reasons. Accordingly, allowing for a grace period
  quantifies an effect that is arguably of more practical relevance.
\end{example}

\begin{remark}
  \citet{wen2021intervention} study stochastic interventions defined
  as a random draw from a density that can be expressed as a mixture
  between a known density $\f_t(a_t\mid h_t)$ and the density of the
  observed data $\g_t(a_t\mid h_t)$. The functional identifying the
  mean counterfactual outcome under this stochastic intervention is
  the same as the functional identifying the counterfactual outcome
  under an LMTP defined by
  \begin{equation*}
    \d(a_t,h_t,\varepsilon_t)=\one\{\varepsilon_t\leq c_t(h_t)\} q_t +
    \one\{\varepsilon_t> c_t(h_t)\} a_t,
  \end{equation*}
  where $\varepsilon_t$ is a random draw from a uniform distribution
  in $[0,1]$, $0<c_t(h_t)<1$ is a mixture constant, and $q_t$
  represents a draw from $\f_t(a_t\mid h_t)$. Efficient estimators for
  this type of LMTP have been previously proposed
  by~\cite{diaz2021lmtp}. \citet{wen2021intervention} discuss the
  stochastic interventions version of these effects, but incorrectly
  claim that the estimators they develop are not particular cases of
  the methodology of~\cite{diaz2021lmtp}.
\end{remark}

\subsection{Identification of the effect of LMTPs under competing risks}\label{sec:iden}

The following assumptions will be sufficient to prove identification:

\begin{assumptioniden}[Supported treatment intervention]\label{ass:support1}
  If $(a_t,h_t)\in \supp\{A_t,H_t\}$ then
  $(\d(a_t,h_t,\varepsilon_t),h_t)\in \supp\{A_t,H_t\}$ for
  $t\in\{1,\ldots,\tau\}$ and all values of $\varepsilon_t$.
\end{assumptioniden}

\begin{assumptioniden}[Positivity of loss-to-follow-up mechanism]\label{ass:support2}
  $\P\{\g_{C,t}(A_t, H_t)>0\} = 1$ for $t\in\{1,\ldots,\tau\}$.
\end{assumptioniden}

\begin{assumptioniden}[Strong sequential randomization of treatment mechanism]\label{ass:exch1}
  $U_{A,t}\indep (\underline U_{D, {t+1}}, \underline U_{Y,
    {t+1}},\allowbreak \underline U_{L, {t+1}},\allowbreak \underline U_{A, t+1})
  \mid (H_t, R_t=C_t=1)$ for all $t\in\{1,\ldots,\tau\}$.
\end{assumptioniden}

\begin{assumptioniden}[Sequential randomization of loss-to-follow-up
  mechanism]\label{ass:exch2}
  $U_{C,t}\indep (\underline U_{D, {t+1}}, \underline U_{Y,
    t+1}, \underline U_{L, t+1},\allowbreak\underline U_{A, t+1}) \mid (A_t,H_t, R_t=1)$ for
  all $t\in\{1,\ldots,\tau\}$.
\end{assumptioniden}

Assumptions~\ref{ass:support1} and~\ref{ass:support2} are required so
that the interventions considered are supported in the data. In words,
Assumption~\ref{ass:support1} requires that if there is an individual
with treatment level $a_t$ and covariate values $h_t$ who is
event-free and uncensored at time $t$, there must also be an
individual with treatment level $\d(a_t,h_t,\varepsilon_t)$ and
covariate values $h_t$ who is event-free and uncensored at time
$t$. Assumption~\ref{ass:support1} would be violated if the
intervention $\d$ is infeasible at time $t$ for a participant with
covariate history $h_t$. As an example of this, consider our
illustrative application, letting $A_t$ be an indicator of having been
intubated at time $t$, and letting $H_t$ denote a high-dimensional
vector containing all the information collected for a patient up until
time $t$, including information on blood oxygen saturation. Consider a
modified treatment policy that would delay intubation by a number of
days $\delta$ for a large value $\delta$. Note that intubation is a
salvage therapy that is almost always used when a patient reaches a
low oxygen saturation point. Therefore, Assumption~\ref{ass:support1}
may be violated for this LMTP since a patient with very low blood
oxygen saturation who is intubated at time $t$ has a very low chance
of not being intubated in the vicinity of time $t$.

Assumption~\ref{ass:support2} states that for every possible value of
the covariates $h_t$ that has positive probability of occurring, there
are uncensored individuals.
Assumption~\ref{ass:exch1} is required to identify the effect of an LMTP. It
states that, among event-free uncensored individuals, $H_t$ contains all the
common causes of treatment at time $t$ and future events, treatments, and
covariates. Assumption~\ref{ass:exch2} is required to identify the effect of
a static intervention that prevents loss to follow-up. It states that, among
event-free uncensored individuals, $H_t$ contains all the common causes of
censoring at time $t$ as well as future events and covariates.

\begin{proposition}[Identification of the effect of LMTPs with
  time-to-event outcomes subject to competing events]\label{theo:iden}
  Set $\q_{\tau+1}= Y_{\tau+1}$, $R_{\tau+1}=1$ \corr{and $Z_{\tau+1}=0$}. For $t=\tau,\ldots,1$,
  recursively define
  \begin{equation}
    \corr{\q_t:
    (a_t, h_t)
    \mapsto \E\left[R_{t+1}\times\q_{t+1}(A_{t+1}^\d,
      H_{t+1}) + Z_{t+1}\mid C_t= R_t=1, A_t=a_t,
      H_t=h_t\right]},\label{eq:seqreg}
  \end{equation}
  and define $\theta = \E[\q_1(A_1^\d, L_1)]$. Under Assumptions
  \ref{ass:support1}-\ref{ass:exch2}, $\P[Y_{\tau+1}(\bar A^\d)=1]$ is identified as
  $\theta$.
\end{proposition}
The identification result above, proved in the supplementary materials, extends
the identification result of~\citet{young2020causal} to LMTPs, and the
identification result of~\cite{diaz2021lmtp} to the case of competing events.
This result allows researchers to estimate from data the counterfactual
probability $\P[Y_{\tau+1}(\bar A^\d)=1]$, interpreted as the outcome rate in
a hypothetical world where the modified treatment policy $\d$ was implemented.

Next, we present efficiency theory and efficient estimators for
$\theta$, based on an extension of our prior work~\citep{diaz2021lmtp}. These
estimators allow the use of slow-converging data-adaptive regression techniques
to obtain estimators of $\theta$ that converge at parametric rates. The
general approach involves finding a von-Mises-type
approximation~\citep{mises1947asymptotic, vanderVaart98,
robins2009quadratic, Bickel97} for the parameter $\q_{L, t}$, which can be
intuitively understood to be a first-order expansion with second-order error
(remainder) terms. As the errors in the expansion are second-order, the
resulting estimator of $\theta$ will be $\sqrt{n}$-consistent and asymptotically
normal as long as the second-order error terms converge to zero at rate
$\sqrt{n}$. This would be satisfied, for example, if all the regression
functions used for estimation converge at rate $n^{1/4}$ (see
\S\ref{sec:estima}). The reader interested in this general theory is encouraged
to consult the works of~\citet{buckley1979linear,rubin2007doubly,
Robins00,Robins&Rotnitzky&Zhao94,vanderLaan2003,Bang05, vdl2006targeted,
vanderLaanRose11, vanderLaanRose18, luedtke2017sequential,
rotnitzky2017multiply}.

\section{Efficient estimation via sequentially doubly robust
  regression}\label{sec:estima}

As noted by~\citet{diaz2021lmtp}, non-parametric $\sqrt{n}$-consistent
estimation is not possible for arbitrary functions $\d$ when the treatment is
continuous. Thus, we impose the following assumption, which is a generalization
of an assumption first used by~\citet{Haneuse2013}. In what follows, we will
assume that the function $\d$ does not depend on the distribution of the
observed data $\P$.

\begin{assumption}[Piecewise smooth invertibility]\label{ass:inv}
  Assume that the cardinality of $A_t$ is $p$. For each $h_t$, assume
  the support of $A_t$ conditional on $H_t=h_t$ may be partitioned
  into $p$-dimensional subintervals
  ${\cal I}_{t,j}(h_t):j = 1, \ldots, J_t(h_t)$ such that
  $\d(a_t, h_t, \varepsilon_t)$ is equal to some
  $\d_j(a_t, h_t, \varepsilon_t)$ in ${\cal I}_{t,j}(h_t)$ and
  $\d_j(\cdot,h_t, \varepsilon_t)$ has inverse function
  $\b_j(\cdot, h_t, \varepsilon_t)$ with Jacobian
  $\b_j'(\cdot, h_t, \varepsilon_t)$ with respect to $a_t$.
\end{assumption}

In what follows, we assume that either (i) $A_t$ is a discrete random variable
for all $t$, or (ii) $A_t$ is a continuous random variable and the modified
treatment policy $\d$ satisfies Condition~\ref{ass:inv}. Next, for continuous
$A_t$ we define
{\footnotesize\begin{equation}\label{eq:gdelta}
    \g_{A,t}^\d(a_t \mid h_t) =
    \sum_{j=1}^{J_t(h_t)} \int \one_{t, j} \{\b_j(a_t, h_t, \varepsilon_t),
    h_t\}\times \g_{A,t}\{\b_j(a_t, h_t, \varepsilon_t)\mid h_t\}
    \times\big|\det\{\b_j'(a_t,h_t, \varepsilon_t)\}\big|\dd\P(\varepsilon_t),
  \end{equation}}%
where $\one_{t,j} \{u, h_t\} = 1$ if $u \in {\cal I}_{t, j}(h_t)$ and
$\one_{t,j} \{u, h_t\} = 0$ otherwise. Under Condition~\ref{ass:inv},
application of the formula for the density of a transformation shows that the
p.d.f.~of $A_t^\d$ conditional on the history $h_t$ is $\g_{A,t}^\d(a_t \mid
h_t)$. In the case of Example~\ref{ex:mshift}, the post-intervention
p.d.f.~becomes
\begin{equation*}
  \g_{A,t}^\d(a_t\mid h_t) = \g_{A,t}(a_t\mid h_t) \one\{a_t \delta <
  l_t(h_t)\} + \delta^{-1}\g_{A,t}(a_t\delta^{-1}\mid h_t)
  \one\{a_t\geq l_t(h_t)\},
\end{equation*}
which shows that piecewise smoothness is sufficient to handle interventions such
as~\eqref{eq:mshift} which are not smooth in the full range of the treatment.
Condition~\ref{ass:inv} and expression~\eqref{eq:gdelta} ensure that we can use
the change of variable formula when computing integrals over $\mathcal A_t$ for
continuous treatments. This is useful for studying properties of the parameter
and estimators we propose.

For discrete treatments, we let
\begin{equation}
  \g_{A,t}^\d(a_t\mid h_t) = \sum_{s_t\in \mathcal A_t}\int
  \one\{\d(s_t, h_t,\varepsilon_t)=a_t\}\g_{A,t}(s_t\mid
  h_t)\dd\P(\varepsilon_t).\label{eq:disc}
\end{equation}
Evaluated using the definition of the incremental propensity score intervention
based on the risk ratio given in Example~\ref{ex:ipsi}, this expression yields
\[
  \g_{A,t}^\d(a_t\mid h_t) = a_t\delta \g_{A,t}(1\mid h_t) + (1-a_t)[1-\delta
    \g_{A,t}(1\mid h_t)],
\]
which clarifies that the risk ratio for receiving the treatment under the
intervention versus the observed data is equal to $\delta$. In what follows, it
will be useful to define the parameter
\[
  \w_t(c_t, a_t, h_t) = \frac{\g_{A,t}^\d(a_t\mid
    h_t)}{\g_{A,t}(a_t\mid h_t)}\frac{c_t\times r_t}{\g_{C,t}(a_t, h_t)},
\]
and the functions
\begin{equation}
  \label{eq:eifgamma}
  \corr{\D_t : x\mapsto
  \sum_{s=t}^\tau\left(\prod_{k=t}^s \w_k(c_k, a_k, h_k)\right)
  \{r_{s+1}\q_{s+1}(a_{s+1}^\d, h_{s+1}) +z_{s+1}-
  \q_s(a_s, h_s)\} + \q_t(a_t^\d, h_t)}
\end{equation}
for $t=\tau,\ldots,1$ which map $x=(d_t, y_t, l_t, a_t,
c_t:t\in\{1,\ldots,\tau\})$ to the real line. When necessary, we use the
notation $\D_t(x;\eta)$ or $\D_t(x;\underline\eta_t)$ to highlight the
dependence of $\D_t$ on $\underline\eta_t=(\w_t,\q_t,\ldots,\w_\tau, \q_\tau)$
. We also use $\eta$ to denote $(\w_1,\q_1,\ldots,\w_\tau, \q_\tau)$, and define
$\D_{\tau+1}(X;\eta) = Y$.

Estimation of $\theta$ will proceed by regression of the transformation
$\D_{t+1}$ sequentially from $t=\tau$ to $0=1$ to obtain estimates $\q_t$. We
say that $\D_{t+1}$ is a sequentially doubly robust unbiased transformation for $\q_t$ due
to the following proposition:

\begin{proposition}[Sequentially doubly robust transformation]\label{prop:vonm}
  Let $\eta'$ be such that either $\q'_s=\q_s$ or $\w'_s=\w_s$ for all
  $s>t$. Then, in the event $R_t=1$, we have
  \[\corr{\E\big[R_{t+1}\times \D_{t+1}(X;\eta') + Z_{t+1}\mid
    C_t=R_t=1,A_t=a_t,H_t=h_t\big] =\q_t(a_t,h_t).}\]
\end{proposition}
This proposition is a straightforward application of the von-Mises expansion
given in Lemma~\ref{lemma:vm} in the supplementary materials. Another
consequence of Lemma~\ref{lemma:vm} is that the efficient influence
function for estimating $\theta = \E[\q_1(A^\d,L_1)]$ in the non-parametric
model is given by $\D_1(X)-\theta$~\citep[see][]{diaz2021lmtp}.
Proposition~\ref{prop:vonm} motivates the construction of the sequential
regression estimator by iteratively regressing an estimate of the data
transformation $\D_{t+1}(X;\eta)$ on $(A_t,H_t)$ among individuals with
$C_{t-1}=R_t=1$, starting at $\D_{\tau+1}(X;\eta)=Y$.
\begin{remark}[Double robustness of incremental propensity
  score interventions]\label{remark:ipsi} Consider the odds ratio
  incremental propensity score interventions of Example~\ref{ex:ipsi},
  where we let $\varepsilon_t\sim U(0,1)$, and define $\d(h_t,
  \varepsilon_t)=\one\{\varepsilon_t< \g_t^\d(1\mid h_t)\}$. The
  result in Proposition~\ref{prop:vonm} also holds for this
  intervention, but it does not allow for doubly robust estimation. To
  see why, notice that $\q_t$ itelse depends on $\{\g_s:s>t\}$ (or
  equivalently, on $\{\w_s:s>t\}$) through $\d$. Therefore, if $\q'_s=
  \q_s$ and $\w'_s\neq \w_s$ for all $s>t$, we get
  \corr{$\E\big[R_{t+1}\times \D_{t+1}(X;\eta') + Z_{t+1}\mid
    C_t=R_t=1,A_t=a_t,H_t=h_t\big] =\q_t(a_t,h_t)$}, where $\q_t$ is
  evaluated at the misspecified $\{\w_s':s>t\}$. This is aligned
  with the findings of \cite{kennedy2019nonparametric}. In contrast,
  the risk ratio incremental propensity score intervention does allow
  for doubly robust estimation, since $\d_t$ does not depend on
  $\w_t$.
\end{remark}
In order to avoid imposing entropy conditions on the initial
estimators, we use sample splitting and
cross-fitting~\citep{klaassen1987consistent, zheng2011cross,
  chernozhukov2018double}. Let ${\cal V}_1, \ldots, {\cal V}_J$ denote
a random partition of the index set $\{1, \ldots, n\}$ into $J$
prediction sets of approximately the same size. That is,
${\cal V}_j\subset \{1, \ldots, n\}$;
$\bigcup_{j=1}^J {\cal V}_j = \{1, \ldots, n\}$; and
${\cal V}_j\cap {\cal V}_{j'} = \emptyset$. In addition, for each $j$,
the associated training sample is given by
${\cal T}_j = \{1, \ldots, n\} \setminus {\cal V}_j$. We let
$\hat \eta_{j}$ denote the estimator of $\eta$ obtained by training
the corresponding prediction algorithm using only data in the sample
${\cal T}_j$. Further, we let $j(i)$ denote the index of the
validation set containing unit $i$. For preliminary estimates
$\hat\w_{1,j(i)},\ldots,\hat \w_{\tau, j(i)}$, the estimator is
defined as follows:
\begin{enumerate}[label=Step \arabic*:, align=left, leftmargin=*]
\item Initialize  $\D_{\tau+1}(X_i;\underline{\check\eta}_{\tau,j(i)})=
  Y_i$ for $i=1,\ldots,n$.
\item For $t=\tau,\ldots,1$:
  \begin{itemize}
  \item Compute the pseudo-outcome
    \corr{$\check Y_{t+1,i} = R_{t+1}\times
    \D_{t+1}(X_i;\underline{\check\eta}_{t,j(i)})+Z_{t+1}$} for all
    $i=1,\ldots,n$.
  \item For $j=1,\ldots,J$:
    \begin{itemize}
    \item Regress $\check Y_{t+1,i}$ on $(A_{t,i,}H_{t,i})$ using any
      regression technique and using only data points
      $i\in \mathcal T_{j}$.\label{step:2}
    \item Let $\check \q_{t,j}$ denote the output of the above regression, update
      $\underline{\check\eta}_{t, j} = (\hat\w_{t,j}, \check
      \q_{t,j},\ldots,\hat\w_{\tau,j}, \check \q_{\tau,j})$, and
      iterate.
    \end{itemize}
  \end{itemize}
\item Define the sequential regression estimator as
  \[\hat\theta=\frac{1}{n}\sum_{i=1}^n\D_1(X_i;\check\eta_{j(i)}).\]
\end{enumerate}
Implementation of the above estimator requires an algorithm for
estimation of the parameters $\w_t$ for all $t$. This parameter may be
estimated component-wise as follows. First, the censoring mechanism
$\g_{C,t}$ may be estimated through any regression or classification
procedure by regressing the indicator $C_t$ on data $(A_t,H_t)$ among
individuals with $R_t=C_t=1$. Thus, it only remains to estimate the
density ratio $\g_t^\d(a_t\mid h_t)/\g_t(a_t\mid h_t)$. An option to
estimate this density ratio is to first obtain an estimate of the
density $\g_t(a_t\mid h_t)$, and then to obtain an estimator of
$\g_t^\d(a_t\mid h_t)$ by applying formula~\eqref{eq:gdelta}
or~\eqref{eq:disc}. This approach is problematic in the case of
continuous or multivariate treatments due both to the curse of
dimensionality and to the fact that flexible and computationally
efficient estimation of conditional densities is an under-developed
area of machine learning (relative to regression), though some suitable
techniques do exist~\citep[e.g.,][]{diaz2011super, hejazi2021haldensify,
hejazi2022efficient}. This approach can also be problematic for discrete
treatments of high cardinality. Since the conditional densities themselves are
unnecessary for estimation, and direct estimation of the ratio $\g_t^\d(a_t\mid
h_t)/\g_t(a_t\mid h_t)$ suffices, we propose to use~\citet{diaz2021lmtp}'s
approach to direct estimation of density ratios. Briefly, this approach starts
by constructing an augmented artificial dataset of size $2n$ in which each
observation is duplicated. For each observation, one duplicate gets assigned the
post-intervention treatment value $\d(A_t, H_t)$ and the other gets the actual
observed value $A_t$. Then, the task of density ratio estimation reduces to the
task of classifying which record belongs to the intervened-upon unit, based on
the treatment and the history $H_t$. Specifically,~\citet{diaz2021lmtp} show
that the odds of this classification problem is equal to the density ratio.
Importantly, data-splitting procedures such as cross-validation or cross-fitting
must be performed such that the two records from the same participants are
always kept in the same fold.

The above estimation algorithm requires the implementation of several regression
procedures. Many algorithms from the machine learning literature are available
for such regression problems, including those based on ensembles of regression
trees, splines, etc. As stated in the asymptotic normality result below, it is
important that the regression procedures used approximate as well as possible
the relevant components of the true data-generating mechanism. Increasing the
predictive performance of these regression fits can entail estimator selection
from a large list of candidate estimators. If the sample size is large enough,
a principled solution to this problem is the use of Super
Learning~\citep{vanderLaanDudoitvanderVaart06, vanderLaanPolleyHubbard07}, a
form of model stacking~\citep{Wolpert1992,Breiman1996} that constructs an
ensemble model as a convex combination of candidate algorithms from
a user-supplied regression library, where the weights in the combination are
computed based on cross-validation and are such that they minimize the
holdout prediction error of the resulting estimator.

The asymptotic distribution of $\theta$ is Gaussian with variance equal to the
efficiency bound, under conditions. To state the conditions, we first define the
data-dependent parameter
\[
  \corr{\check\q_t^\dagger(a_t, h_t) = \E\big[R_{t+1}\D_{t+1}
  (X;\underline{\check\eta}_t) +Z_{t+1}\mid C_t=R_t=1, A_t=a_t,H_t=h_t\big]},
\]
where the outer expectation is only with respect to the distribution $\P$ of $X$
(i.e., $\check\eta$ is fixed). The following theorem is a consequence of
Theorem 4 of~\citet{diaz2021lmtp}:
\begin{theorem}[Consistency and weak convergence of SDR
  estimator]\label{theo:asrem}
  We have
  \begin{enumerate}[label=(\roman*)]
  \item If, for each time $t$, either $||\hat\w_t - \w_t|| = o_\P(1)$
    or $||\check\q_t - \check\q^\dagger_t|| = o_\P(1)$, then
    $\hat\theta =\theta + o_\P(1)$.
  \item If
    $\sum_{t=1}^\tau||\hat\w_t - \w_t||\, ||\check\q_t -
    \check\q_t^\dagger|| = o_\P(n^{-1/2})$ and
    $\P\{\w_t(A_t,H_t) < c\}=\P\{\hat\w_t(A_t,H_t) < c\}=1$ for some
    $c<\infty$, then
    \[\sqrt{n}(\hat\theta - \theta) =
      \frac{1}{\sqrt{n}}\sum_{i=1}^n\D_1(X_i;\eta) + o_P(1),\]
    implying that
    $\sqrt{n}(\hat\theta - \theta) \rightsquigarrow N(0,\sigma^2)$,
    where $\sigma^2=\var\{\D_1(X;\eta)\}$ is the non-parametric
    efficiency bound.
  \end{enumerate}
\end{theorem}
Note that this means that $\hat\theta$ is \textit{sequentially
doubly robust} in the sense that it is consistent if at least one of
the two nuisance parameters is consistently estimated at each time
point. This property is also known in the literature as
\textit{$2^\tau$-multiple robustness}. We prefer the term \textit{sequentially
  doubly robust} for two reasons. First, contrasted to the double
robustness property which requires one out of two nuisance estimators
to be consistent, the term $2^\tau$-multiple robustness erroneously
conveys the message that consistency would follow from one out of
$2^\tau$ nuisance estimators being consistent. Second, retaining the
term ``double'' in this description conveys a fundamental property of
the estimation procedure: it is based on an expansion (given in
Lemma~\ref{lemma:vm} in the supplementary material) with
second-order remainder terms. In this nomenclature, a triply robust
estimating procedure would be one based on an expansion with
third-order remainder terms, and so forth.\footnote{Some of these points were
brought to our attention by Edward H.~Kennedy in a conversation on Twitter
(see this thread,
\url{https://twitter.com/edwardhkennedy/status/1446952129218949128?s=20&t=8o_-20szw0PPfqocYiw7Dg}).}

A drawback of the above estimation approach is that there is no
guarantee that the parameter will remain within the bounds of the
parameter space. That is, it is possible, in principle, that the
sequential regression estimator will provide time-to-event
probabilities outside of the closed unit interval $[0,1]$. In the
supplementary materials, we present a targeted minimum loss estimator
(TMLE) capable of producing estimates that always lie within the unit
interval.  However, boundedness of the TMLE comes with a robustness
tradeoff. In particular, the robustness properties of the TMLE are
inferior to those of the sequentially doubly robust
estimator. Specifically, consistency of the TMLE requires that either
$\lVert \hat\w_t - \w_t \rVert = o_\P(1)$ or
$\lVert \hat\q_t - \q_t \rVert = o_\P(1)$, for an estimator
$\hat\q_t$ constructed based on sequential regression
using~\eqref{eq:seqreg}. As a result, consistent estimation of $\q_t$
in the TMLE procedure requires consistent estimation of $\q_s:
s>t$. By contrast, the sequentially doubly robust estimator provides
consistent estimation of $\q_t$ when either $\q_s$ or $\w_s$ is
consistently estimated for all
$s>t$. Both~\citet{luedtke2017sequential} and \citet{diaz2021lmtp}
provide further in-depth discussion on this topic. Furthermore, the
conditions required for $\sqrt{n}$-consistency for this TMLE algorithm
may be stronger than those required for the SDR estimator. In the supplementary
materials, we provide an argument that the second-order term associated to the
TMLE is bounded above by the term
\[
  \sum_{t=1}^\tau\left(\lVert \hat\w_t - \w_t\rVert\sum_{s=t}^{\tau} \lVert
  \hat\q_s - \q_s^\star \rVert\right),
\]
where $\q_t^\star$ is the data-dependent parameter defined as
\[
  \corr{\q_t^\star(a_t, h_t) = \E\big[R_{t+1}\hat \q_{t+1}(a_{t+1},
    h_{t+1}) + Z_{t+1}\mid
  C_t=R_t=1, A_t=a_t,H_t=h_t\big]}.
\]
In comparison to the assumption in statement (ii) of the above
theorem, convergence of the second-order term associated to the TMLE
seems to require stronger assumptions that would also require
convergence of all the cross-product terms $\lVert \hat\w_t - \w_t
\rVert\lVert \hat\q_s - \q_s^\star \rVert$ for $s\geq t$.

Neither the SDR nor the TMLE guarantee that a survival function
estimated at several points in time will result in monotonic
decreasing estimates. In order to leverage the desirable robustness
properties of the SDR estimator while endowing it with the properties
of boundedness and monotonicity, we leverage a set of techniques
proposed by~\citet{westling2020correcting}. Under their approach,
out-of-bounds estimates are truncated to remain within bounds, and the
possibly non-monotonic estimate of the survival curve, alongside any
confidence region limits, is projected onto the space of monotonic
functions by way of isotonic regression. Importantly,
\citet{westling2020correcting} show that the resulting projected
estimator retains the same asymptotic properties as the original
estimator. In our case, this implies that the projected estimator
retains the properties of asymptotic normality and consistency under
the assumptions of Theorem~\ref{theo:asrem}. Specifically,
simultaneous confidence bands centered around the isotonic projection
constructed using the multivariate distribution implied by
Theorem~\ref{theo:asrem} will have correct asymptotic
coverage under the assumptions of the theorem.

Under the conditions of the above theorem, a Wald-type procedure, in which the
standard error may be estimated based on the empirical variance of
$\D_1(X_i;\check\eta_{j(i)})$, yields asymptotically valid confidence intervals
and hypothesis tests. Furthermore, the asymptotic linearity of
Theorem~\ref{theo:asrem} implies that estimates of survival functions for
multiple time points asymptotically converge to a multivariate normal
distribution, which can be used to construct simultaneous confidence intervals
using the method described, e.g., by~\cite{hothorn2008simultaneous}.

\section{Illustrative application}

\subsection{Background on motivating study and data}

To illustrate the proposed methods, we aim to answer the following question:
\textit{what is the effect of invasive mechanical ventilation (IMV) on
acute kidney injury (AKI) among COVID-19 patients?}  This question has been
recently identified by an expert panel on lung--kidney interactions in critically
ill patients~\citep{joannidis2020lung} as an area in need of further research.
AKI is a common condition in the ICU, complicating about 30\% of ICU admissions,
and causing increased risk of in-hospital mortality and long-term morbidity and
mortality~\citep{kes2010acute}.
AKI is often viewed as a tolerable consequence of interventions to support other
failing organs, such as IMV~\citep{husain2016lung}.

To answer this question, we will use data on approximately 3,300 patients
hospitalized with COVID-19 at the New York Presbyterian Cornell, Queens, or
Lower Manhattan hospitals. Our analysis focuses on patients hospitalized between
March 3\textsuperscript{rd} and May 15\textsuperscript{th}, 2020, who did not
have a previous history of Chronic Kidney Disease (CKD). This choice of
timeframe is rooted in there being a lack of treatment guidelines in the initial
months of the COVID-19 pandemic, which resulted in a large variability in
provider practice regarding mechanical ventilation. This clinical heterogeneity
in the deployment of IMV makes the problem particularly well-suited to
assessment via the causal effects of LMTPs. The analytical dataset was built
from two distinct databases. Demographics, comorbidity, intubation, death, and
discharge data were abstracted from electronic health records by trained medical
professionals into a secure RedCAP database~\citep{goyal2020clinical}. These
data were supplemented with the Weill Cornell Critical carE Database for
Advanced Research (WC-CEDAR), a comprehensive data repository housing
laboratory, procedure, diagnosis, medication, microbiology, and flowsheet data
documented as part of standard care~\citep{schenck2021critical}. This dataset is
unique in that it contains a large sample of patients treated in the same
hospital system within a short period of time, while still exhibiting extensive
treatment variability.

The most frequent reason for morbidity and mortality among patients with
COVID-19 is pneumonia leading to acute respiratory distress syndrome
(ARDS). Respiratory support via devices such as nasal cannulae, face
masks, or endotracheal tubes (i.e., intubation and subsequent
mechanical ventilation) is a primary treatment for
ARDS~\citep{hasan2020mortality}. For severe ARDS patients, an
important decision must be made regarding the optimal timing of
intubation based on clinical factors such as oxygen saturation,
dyspnea, respiratory rate, chest radiograph, and, occasionally,
external resource considerations such as ventilator or provider
availability~\citep{tobinresp, thomson2021timing}. At the outset of
the COVID-19 pandemic, multiple international guidelines recommended
early intubation in an effort to protect healthcare workers from
infection and to avoid complications stemming from ``crash''
intubations~\citep{papoutsi2021effect}. Yet, as the pandemic
progressed and multiple reports documented high mortality for
mechanically ventilated patients, guidance changed to delaying
intubation with the rationale that IMV increases risk of
secondary infections, ventilator-induced lung injury, damage to other
organs, and death~\citep{bavishi2021timing, tobin2006principles}. The
lungs and kidneys are physiologically closely tied, and it has been
hypothesized that mechanical ventilation may cause AKI in COVID-19
patients via oxygen toxicity and capillary endothelial damage leading
to inflammation, hypotension, and
sepsis~\citep{durdevic2020progressive}.

Despite these hypotheses, a causal link between intubation and AKI in patients
with COVID-19 has not been definitively established. Furthermore, such a link is
difficult to study in a randomized fashion due to the clinically subjective
nature of intubation assignment --- that is, it is impossible to pick an oxygen
saturation breakpoint to safely intubate across all patients and hospital
settings~\citep{tobin2020caution, perkins2020recovery}.

While the original dataset cannot be made publicly available due to patient
privacy considerations, we have created a synthetic dataset that mimics the
general structure of the real clinical dataset, and we have made this synthetic
dataset available publicly. The \texttt{R} code and the synthetic dataset are
available in the GitHub repository at
\url{https://github.com/kathoffman/lida-comprisks}.

\subsection{Time-varying variables and modified treatment policy}

For each patient, study time begins on the day of hospitalization. The
treatment of interest is daily maximum respiratory support, which can
take three levels: no supplemental oxygen, supplemental oxygen not
including IMV, and IMV. Our goal is to estimate the overall causal
effect on daily AKI rates of an intervention that would delay
intubation for IMV by a single day among 449 patients who received IMV
for more than one day, and a prevention of IMV for patients who
received IMV for exactly one day. To simplify the presentation, in
what follows we refer to this intervention as a ``one-day delay in
intubation.'' Death by other causes is treated as a competing risk for
AKI development. In notation, this LMTP may be expressed as
follows: Let $A_t\in\{0,1,2\}$, where $0$ indicates no supplemental
oxygen at end of day $t$, $1$ indicates supplemental oxygen not
including IMV, and $2$ indicates IMV. Consider the following
intervention:
\begin{equation}\label{eq:exmtp}
  \d_t(a_t,h_t) =
  \begin{cases}
    1 &\text{ if } a_t=2 \text{ and } a_s \leq 1 \text{ for all } s < t,\\
    a_t & \text{ otherwise.}
  \end{cases}
\end{equation}
Under this intervention, a patient who is first intubated at day $t$
would receive non-invasive oxygen support at day $t$. Otherwise, the
patient would receive the oxygen support actually (i.e.,
``naturally'') observed at day $t$. This intervention assesses what
would have happened in a hypothetical world where intubation
procedures were more conservative (operationalized as a one-day delay)
than they were at the beginning of the pandemic in New York City's
surge conditions, when limited information and treatments were
available for COVID-19.  Figure~\ref{fig:os} shows the observed
supplemental oxygen levels and observed AKI and death rates across
days 1-14.

Baseline confounders include age, sex, race, ethnicity, body mass
index (BMI), hospital location, home oxygen status, and comorbidities
(e.g., hypertension, history of stroke, Diabetes Mellitus, Coronary
Artery Disease, active cancer, cirrhosis, asthma, Chronic Obstructive
Pulmonary Disease, Interstitial Lung Disease, HIV infection, and
immunosuppression), and are included in $L_1$.  Time-dependent
confounders include vital signs (e.g., highest and lowest respiratory
rate, oxygen saturation, temperature, heart rate, and blood pressure),
laboratory results (e.g., C-Reactive Protein, BUN Creatinine Ratio,
Activated Partial Thromboplastin time, Creatinine, lymphocytes,
neutrophils, bilirubin, platelets, D-dimer, glucose, arterial partial
pressure of oxygen, and arterial partial pressure of carbon dioxide),
and an indicator for daily corticosteroid administration greater than
0.5 mg/kg body weight. Time-dependent confounders were measured
irregularly for each patient. We bin them in 24-hour intervals from
the time of hospitalization. In cases of multiple measurements in a
24-hour window, the clinically worst measure was used. In cases of
missing baseline confounders, multiple imputation using chained
equations (MICE)~\citep{mice2011jss} is used. For missing data at
later time points, the last observation is carried forward. Patients
are censored at their day of hospital discharge, as AKI and vital
stataus were unknown after this point. A distribution of the censoring
weights is given in Figure~\ref{fig:cens} in the supplementary
materials.

A possible limitation of our setup is that we treat discharge as a
censoring event. If patients who remain in-hospital at the end of
follow-up are very different from those who are discharged, and if
such differences are not explained by measured covariates, then our
results may be biased. Specifically, our analyses rely on
Assumption~\ref{ass:exch2}, which applied to this problem states that
$(A_t, H_t)$ contains all the common causes of discharge at time $t$
and future variables measured at times $s>t$, including future
outcomes. This assumption could be violated if there is an unmeasured
common cause of discharge and outcomes such as an underlying
biological or clinical process not captured by the covariates listed
in the above paragraph. 

\begin{figure}[!htp]
  \centering
  \includegraphics[scale = 0.5]{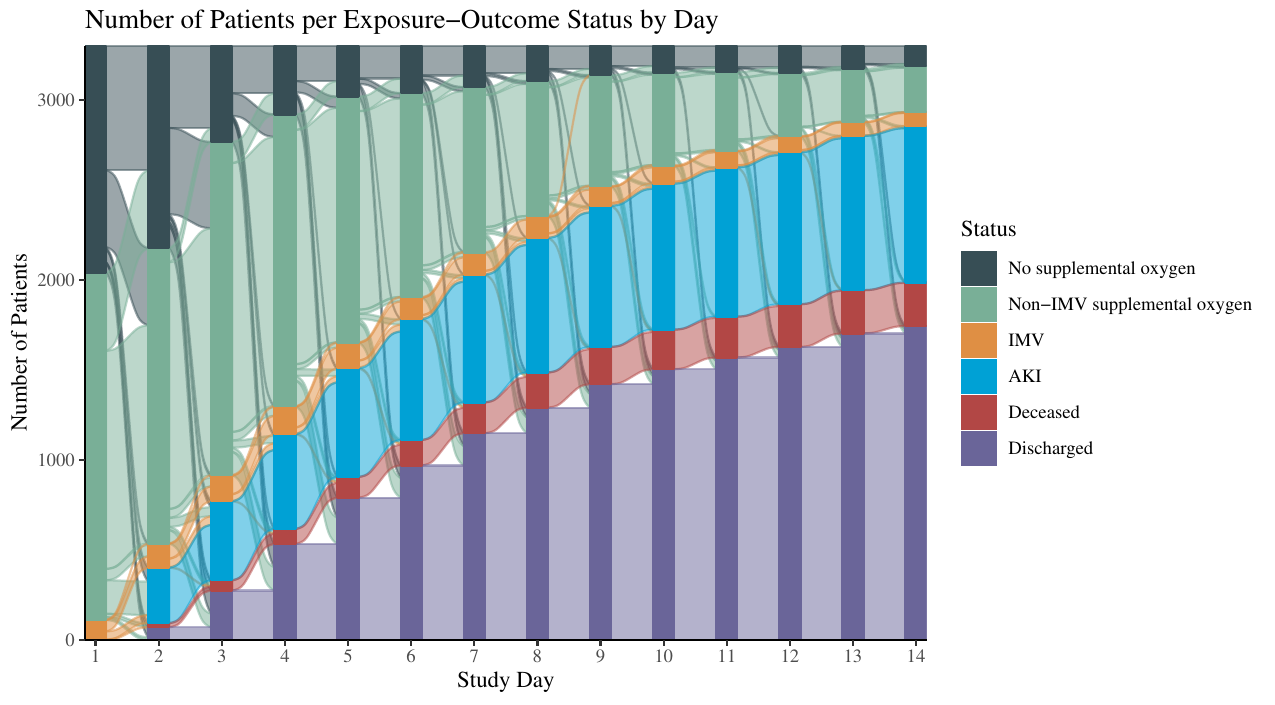}
  \caption{A flow diagram, or ``alluvial plot", showing the movement
    of patients between exposure and outcome statuses. The $x$ axis
    shows study day and the $y$ axis shows counts of patients. Colors
    indicate exposure status for patients still in the study, or
    outcome status for patients who have already experienced the
    outcome, competing risk, or loss-to-follow-up. Darker-colored bars
    indicate the number of patients in each exposure or outcome at
    that time point, while lighter-colored bars visualize patients
    moving between time points. The exposure contains three levels: no
    supplemental oxygen, supplemental oxygen not including IMV, and
    IMV. The outcome of interest is AKI, and death is a competing
    risk. Patients who are discharged are right-censored. There is a
    one-day lag (i.e.  no exposure status shown for the last study
    day) so that all data used in the analysis can be visualized on a
    single plot.}
  \label{fig:os}
\end{figure}

\subsection{Results}

We estimate the effect of the LMTP with the SDR estimator (on account of its
enhanced robustness relative to TMLE), using a version of the \texttt{lmtp}
\texttt{R} package~\citep{williams2022lmtp} modified to support the competing
risks setting. The nuisance functions $\w_t$ and $\q_t$ are estimated using the
Super Learner ensemble modeling algorithm~\citep{vanderLaanPolleyHubbard07} via
the \texttt{sl3} \texttt{R} package~\citep{coyle2022sl3}. The candidate library
considered by the Super Learner included a total of 32 learning algorithms,
inclusive of hyperparameter variations. The algorithms included multivariate
adaptive regression splines~\citep[MARS;][]{friedman1991multivariate};
random forests~\citep{breiman2001random, wright2017ranger}; extreme gradient
boosted trees~\citep{friedman2001greedy, chen2016xgboost}; lasso, ridge, and
elastic net (with equal weights, i.e., $\alpha = 0.5$) regularized
regression~\citep{friedman2009glmnet}; main terms generalized linear
models~\citep{mccullagh1989generalized, enea2009fitting}; generalized linear
models with Bayesian priors on parameter estimates~\citep{gelman2006data}; and
an unadjusted outcome mean.  The estimators of $\w_t$ and $\q_t$ were computed
under a Markov assumption with lag of two days to reflect the collection of laboratory
results at minimum 48-hour intervals.


Note that a complete clinical understanding of the effect of IMV would
require estimating its effect on death, as these interventions are
meant to save lives. Figure~\ref{fig:sdr_survest} displays the
SDR-based incidence difference between the two treatment conditions
for both outcomes, together with 95\% simultaneous confidence bands.

\begin{figure}[!htb]
  \centering
  \includegraphics[scale=0.5]{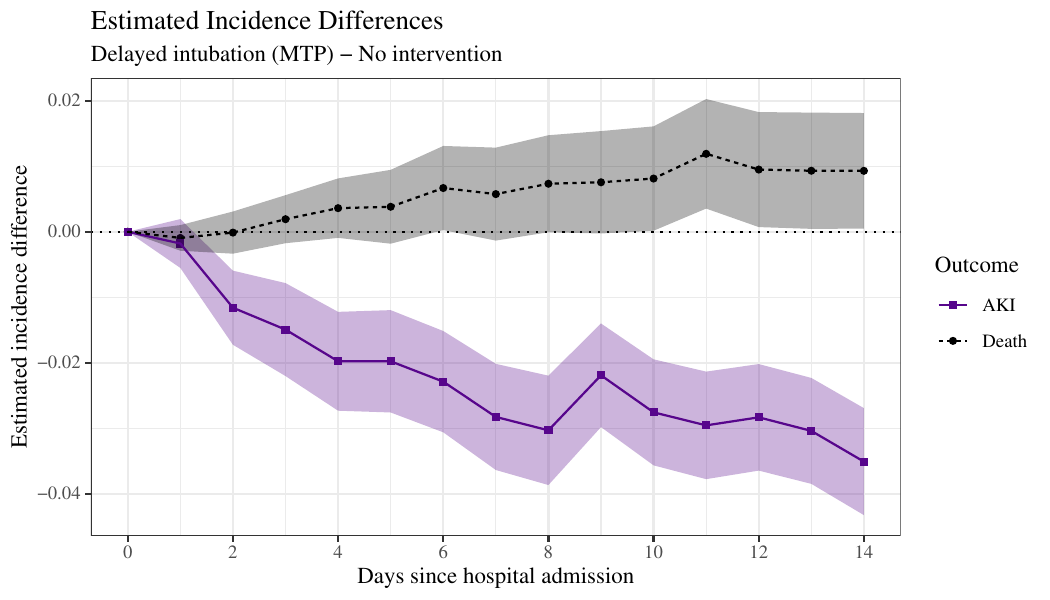}
  \caption{Difference between the two treatment policies for two
    outcomes: AKI and mortality. 95\% simultaneous confidence bands
    cover the sets of point estimates.}
  \label{fig:sdr_survest}
\end{figure}
Inspection of Figure~\ref{fig:sdr_survest} makes apparent the
protective effect of the delay-in-intubation policy against AKI. The
estimates of the incidence difference are consistently negative and
decrease steadily in time; moreover, the simultaneous confidence bands
do not overlap with the null value of zero except for on the first
day. Hypothesis testing of the incidence difference yields adjusted
p-values $< 0.001$ for all days after the first, indicating a
statistically significant protective effect of the LMTP; the
Bonferroni correction was used to adjust for testing
multiplicity. Based on these estimates, the delay-in-intubation policy
would be expected to result in a decrease in 14-day AKI rates by
approximately 3.5\% (simultaneous 95\% CI: [2.7\%, 4.3\%]) relative to
no intervention. Importantly, while IMV does indeed have the effect of
increasing survival, its negative effect on AKI is much larger than
its positive effect on survival. This raises the question of whether
there are patients for whom there is no survival benefit who are being
harmed by the intervention. Developing methods for identifying those
patients is the subject of future methodological research.



\section{Discussion}

Modified treatment policies represent a recent but powerful advance in
modern causal inference. These intervention regimes have a variety of
applications, including (i) flexibility in the definition of effects
whose identification relies on a positivity assumption that may be
satisfied by design; (ii) flexibility in the definition and estimation
of the effects of continuous, time-valued, and multivariate
treatments; and (iii) definition of effects that depend on the natural
value of treatment. Here, we have advanced the LMTP methodology,
extending it to the case of survival outcomes with competing
events. We have presented identification assumptions for the causal
parameters, as well as a sequentially doubly robust and efficient
estimation algorithm capable of leveraging the many data-adaptive
regression procedures available for nuisance parameter estimation. We
ensure that the estimates of the survival function remain within the
parameter space (monotonic decreasing and bounded in the closed unit
interval) by leveraging the monotonic projection framework
of~\citet{westling2020unified} and~\citet{westling2020correcting}. We
have illustrated the scientific utility of our approach through
application to the study of negative side effects of invasive
mechanical ventilation on acute kidney injury in patients hospitalized
with COVID-19, for which death by other causes serves as a competing
event.

Our Proposition~\ref{prop:vonm} shows that doubly robust estimation is possible
for any intervention that is defined in terms of an LMTP satisfying the stated
conditions. Namely, any randomized intervention defined by a known function
$\d(a_t,h_t,\varepsilon_t)$ through either an LMTP or a general stochastic
intervention accommodates doubly robust estimation, provided that: (1) either
(i) $A_t$ is a discrete random variable for all $t$, or (ii) $A_t$ is
a continuous random variable and the modified treatment policy $\d$ satisfies
Condition~\ref{ass:inv}, and (2) the randomizer $\varepsilon_t$ is drawn
independently across units and independently of $U$, and its distribution does
not depend on $\P$. Notably, this rich class contains all interventions
considered by Theorem 2 of~\citet{wen2021intervention}.

A limitation of the approach taken in this paper is that we considered
data collected in discrete time. While the assumption of discrete time
is sensible in many applications (e.g., most clinical measures are
recorded in discrete time, whether the scale is seconds, minutes,
hours, or days), our methods may be of limited applicability when the
time scale is truly continuous or very fine (e.g., seconds) such that
many time points need to be considered. In such cases, it may be
possible to coarsen the time scale into discrete components. Whether such
an approach provides sensible answers and the appropriate level of
coarsening needs to be assessed on a case-by-case basis based on
subject-matter relevance for the application at hand.


\section*{Conflict of Interest}

The authors declare that they have no conflict of interest.

\appendix

\section{Auxiliary results}
The proofs of results in the paper rely heavily on existing
identification and optimality results for general longitudinal
modified treatment policies~\citep{diaz2021lmtp}. We first
present those results. Consider a data structure
$Z=(W_1, E_1, \ldots,W_{\tau}, E_{\tau}, Y)$. We assume these data are
generated by a non-parametric structural equation model with exogenous
errors $U_{W_t}$ and $U_{E_t}$ as detailed by
\cite{diaz2021lmtp}. Let $\m_{C, \tau+1}=Y$,
$V_t = (\bar W_t, \bar E_{t-1})$, and let $\h(e_t, v_t)$ denote a
modified treatment policy on $E$, with
$E_t^\h=\h(E_t(\bar E_{t-1}^\h),V_t(\bar E_{t-1}^\h))$ denoting the
intervened exposure variable. For $t=\tau,\ldots,0$, let
\[\m_t(e_t, v_t)=\E[\m_{t+1}(E_{t+1}^\h, V_{t+1})\mid E_t= e_t,
  V_t=v_t].\] Furthermore, let $\g_t(e_t\mid v_t)$ denote the
probability density function of $E_t$ conditional on $V_t=v_t$,
$\g_t^\h(e_t\mid v_t)$ the probability density function of $E_t^\h$
conditional on $V_t=v_t$, and define
\[\r_t(e_t,v_t) = \frac{\g_t^\h(e_t\mid v_t)}{\g_t(e_t\mid v_t)}.\]
Define
\[\psi_t(z) = \sum_{s=t}^\tau \left(\prod_{k=t}^s\r_k(e_k, v_k)
  \right)\{\m_{s+1}(e_{s+1}^\h,v_{s+1}) - \m_s(e_s,v_s)\} +
  \m_t(e_t^\h,v_t).\]
Then \cite{diaz2021lmtp} proved the following three
results.
\begin{theorem}\label{theo:gcompid}
  Let $Y(\bar E^\h)$ denote a counterfactual outcome in a hypothetical
  world where $\bar E = \bar E^\h$. Assume
  $U_{E,t}\indep \ubar U_{W,t+1}\mid V_t$ and that
  $(e_t,v_t)\in\supp\{E_t,V_t\}$ implies $(\h(e_t,v_t),v_t)\in\supp\{E_t,V_t\}$. Then
  $\lambda=\E[\m_1(E_1^\h,V_1)]$ identifies $\E\{Y(\bar E^\h)\}$.
\end{theorem}

\begin{theorem}\label{theo:eifgcomp}
  The efficient influence function for estimating $\lambda$ in the
  non-parametric model is given by $\psi_1(Z)-\lambda$.
\end{theorem}

\begin{theorem}\label{theo:dr} For any
  fixed sequence
  $(\m_{1}', \r_{1}', \ldots, \m_{\tau}', \r_{\tau}')$,
  define
  \[\C_{t,s}' =
    \prod_{l=t+1}^{s-1}\r_l'(E_l, V_l),\]
  and
  \[\Remi_t(e_t, v_t) =
    \sum_{s=t+1}^\tau\E\left[\C_{t,s}'\{\r_s'(E_s,V_s) - \r_s(E_s,V_s)\}\{\m_s'(E_s,V_s) - \m_s(E_s,V_s)\}\mid
      E_t = e_t, V_t=v_t\right].\]
  Then we have
  \[\m_t(e_t, v_t) = \E[\psi_{t+1}'(Z)\mid E_t=e_t, V_t=v_t] +
    \Remi_t(e_t, v_t)\] where $\psi_{t+1}'$ is defined as above with
  $\m_s$ and $\r_s$ replaced by $\m_s'$ and $\r_s'$, respectively.
\end{theorem}

\section{Identification (Proposition \ref{theo:iden})}
\begin{proof}
  Let $V_t=(D_t, Y_t, L_t)$ and $E_t=(A_t, C_t)$. Let
  $\h(E_t, V_t) = (\d(A_t,H_t), 1)$. This theorem follows from
  application of Theorem~\ref{theo:gcompid} above. First, note that
  for $t=\tau$, we have
  \begin{align*}
    \m_t(e_t^\h,v_t)
    &= \E[Y_{\tau+1}\mid D_t=d_t, Y_t=y_t, H_t=h_t,
      A_t=a_t^\d, C_t=1]\\
    &\corr{= \one\{d_t=0,y_t=1\}\E[Y_{\tau+1}\mid D_t=0, Y_t=C_t=1, H_t=h_t,
      A_t=a_t^\d] + \one\{d_t=1,y_t=1\}}\\
    &\corr{= r_t\q_t(a_t^\d, h_t) + z_t}.
  \end{align*}
  For $t<\tau$, we have
  \begin{align*}
    \m_t(e_t^\h,v_t)
    &= \E[\m_{t+1}(E_{t+1}^\h,V_{t+1})\mid D_t=d_t, Y_t=y_t, H_t=h_t,
      A_t=a_t^\d, C_t=1]\\
    &= \E[R_{t+1}\q_{t+1}(A_{t+1}^\d, H_{t+1}) + Z_{t+1}\mid  D_t=d_t, Y_t=y_t, H_t=h_t,
      A_t=a_t^\d, C_t=1]\\
    &=\corr{ r_t\E[R_{t+1}\q_{t+1}(A_{t+1}^\d, H_{t+1}) + Z_{t+1}\mid R_t=C_t=1, H_t=h_t,
      A_t=a_t^\d]} + z_t\\
    &= \corr{r_t\q_t(a_t^\d, h_t) + z_t}.
  \end{align*}
  Since \corr{$R_1=1$ and $Z_1=0$}, we have $\m_1(E_1^\h,V_1)=\q_1(A_1^\d, H_1)$,
  concluding the proof of the result.
\end{proof}

\section{Efficient influence function}
\begin{proof}
  As before, let $V_t=(D_t, Y_t, L_t)$ and $E_t=(A_t, C_t)$. Let
  $\h(E_t, V_t) = (\d(A_t,H_t), 1)$. First, notice that
  \begin{align*}
    \r_t(e_t, v_t) & = \frac{\g_t^\h(e_t\mid v_t)}{\g_t(e_t\mid
                     v_t)}\\
                   & = \frac{\g_{A,t}^\d(a_t\mid h_t)}{\g_{A,t}(a_t\mid h_t)}\frac{c_t}{\g_{C,t}(a_t,h_t)}.
  \end{align*}
  Then we have
  \begin{align}
    \psi_t(z)
    &= \sum_{s=t}^\tau \left(\prod_{k=t}^s\r_k(e_k,
      v_k)\right)\{\m_{s+1}(e_{s+1}^\h,v_{s+1}) -
      \m_s(e_s,v_s)\} + \m_t(e_t^\h,v_t)\notag\\
    &= \corr{\sum_{s=t}^\tau \left(\prod_{k=t}^s\r_k(e_k,
      v_k)\right)\{r_{s+1}\q_{s+1}(a_{s+1}^\h,h_{s+1}) + z_{s+1}-
      r_s\q_s(a_s,h_s) - z_s\} + r_t\q_t(a_t^\d,h_t) + z_t}\notag\\
    &= \corr{\sum_{s=t}^\tau \left(\prod_{k=t}^s\r_k(e_k,
      v_k)\right)r_s\{r_{s+1}\q_{s+1}(a_{s+1}^\h,h_{s+1}) + z_{s+1} -
      \q_s(a_s,h_s)\} + r_t\q_t(a_t^\d,h_t) + z_t}\notag\\
    &= \corr{\sum_{s=t}^\tau \left(\prod_{k=t}^s\r_k(e_k,
      v_k)r_k\right)\{r_{s+1}\q_{s+1}(a_{s+1}^\h,h_{s+1}) +z_{s+1}-
      \q_s(a_s,h_s)\} + r_t\q_t(a_t^\d,h_t) + z_t}\notag\\
    &= \corr{r_t\left[\sum_{s=t}^\tau \left(\prod_{k=t}^s\w_t(c_k,a_k,
      h_k)\right)\{r_{s+1}\q_{s+1}(a_{s+1}^\h,h_{s+1}) +z_{s+1} -
      \q_s(a_s,h_s)\} + \q_t(a_t^\d,h_t)\right] + z_t}\notag\\
    &=\corr{r_t\varphi_t(x) + z_t}\label{eq:eqeif},
  \end{align}
  \corr{where we used that $r_s=0$ implies $r_{s+1}=0$ and
    $z_s=z_{s+1}$, and $r_s=1$
    implies $z_{s}=0$}. An application of Theorem~\ref{theo:eifgcomp} together with the
  above results for $t=1$ yields the desired result after noticing
  that $R_1=1$.
\end{proof}

\section{von-Mises expansion}
\begin{lemma}[von-Mises expansion]\label{lemma:vm}
  define
  \[\H_{t,s}' =
    \prod_{l=t+1}^{s-1}\frac{\g_{A,t}^{\d'}(A_t\mid H_t)}{\g_{A,t}'(A_t\mid H_t)}\frac{C_t}{\g_{C,t}'(A_t,H_t)},\]
  Let
  {\small\[\Rem_t(a_t,h_t)=\sum_{s=t+1}^\tau\E\left[\H_{t,s}'R_s\{\w_s'(A_s,H_s) -
        \w_s(A_s,H_s)\}\{\q_s'(A_s,H_s) - \q_s(A_s,H_s)\}\mid R_t=C_t=1,
        A_t=a_t, H_t=h_t
      \right].\]}%
  Then, in the event $R_t=1$:
  \[\corr{\q_t(a_t, h_t)=\E[R_{t+1}\varphi_{t+1}'(X) + Z_{t+1}\mid R_t= C_t=1, A_t=a_t, H_t=h_t] + \Rem_t(a_t, h_t).}\]
\end{lemma}
\begin{proof}
  This results follows from an application of Theorem~\ref{theo:dr} to
  a data structure $V_t=(D_t, Y_t, L_t)$ and $E_t=(A_t, C_t)$ with
  $\h(E_t, V_t) = (\d(A_t,H_t), 1)$, together with expression
  (\ref{eq:eqeif}). First, note that
  {\scriptsize
    \begin{align*}
      \Remi_t(e_t, v_t)
      & = \sum_{s=t+1}^\tau\E\left[\C_{t,s}'\{\r_s'(E_s,V_s) - \r_s(E_s,V_s)\}\{\m_s'(E_s,V_s) - \m_s(E_s,V_s)\}\mid
        E_t = e_t, V_t=v_t\right]\\
      & = \sum_{s=t+1}^\tau\E\left[\H_{t,s}'R_s\{\w_s'(A_s,H_s) -
        \w_s(A_s,H_s)\}\{\q_s'(A_s,H_s) - \q_s(A_s,H_s)\}\mid D_t=d_t, Y_t=y_t, C_t=c_t,
        A_t=a_t, H_t=h_t
        \right]\\
      & = \sum_{s=t+1}^\tau\E\left[\H_{t,s}'R_s\{\w_s'(A_s,H_s) -
        \w_s(A_s,H_s)\}\{\q_s'(A_s,H_s) - \q_s(A_s,H_s)\}\mid D_t=d_t, Y_t=y_t, C_t=c_t,
        A_t=a_t, H_t=h_t
        \right]\\
      &= r_t\sum_{s=t+1}^\tau\E\left[\H_{t,s}'R_s\{\w_s'(A_s,H_s) -
        \w_s(A_s,H_s)\}\{\q_s'(A_s,H_s) - \q_s(A_s,H_s)\}\mid R_t=C_t=1,
        A_t=a_t, H_t=h_t
        \right]\\
      &= r_t\Rem_t(a_t,h_t)\\
    \end{align*}}%
  Specifically, we have
  \begin{align*}
    \corr{r_t\q_t(a_t, h_t) +z_t}&= \m_t(e_t, h_t)\\
                      &= \E[\psi_{t+1}'(Z)\mid E_t=e_t, V_t=v_t] + \Remi_t(e_t, v_t)\\
                      &= \corr{\E[R_{t+1}\varphi_{t+1}'(X) + Z_{t+1}\mid D_t=d_t, Y_t=y_t, C_t=c_t,
                        A_t=a_t, H_t=h_t] + \Remi_t(e_t, v_t)}\\
                      &= \corr{r_t\E[R_{t+1}\varphi_{t+1}'(X) + Z_{t+1}\mid R_t=
                        C_t=1, A_t=a_t, H_t=h_t] + z_t + r_t\Rem_t(a_t, h_t)},\\
  \end{align*}
  which concludes the proof.
\end{proof}

\section{Targeted Minimum Loss Based Estimator}

The targeted minimum loss-based estimator is computed as a
substitution estimator that uses an estimate $\tilde\q_{1,j(i)}$
carefully constructed to solve the cross-validated efficient influence
function estimating equation
$\Pn \{\D_1(\cdot, \tilde\eta_{j(\cdot)}) - \thetatmle\}= 0$. The
construction of $\tilde\q_{1,j(i)}$ is motivated by the observation
that the efficient influence function of $\theta$ can be expressed as
a sum of terms of the form:
\[\corr{\left(\prod_{k=1}^t \w_k(c_k,a_k,
      h_k)\right)\{r_{t+1}\q_{t+1}(a_{t+1}^\d, h_{t+1}) + z_{t+1} -
  \q_t(a_t, h_t)\}},\] which take the form of score functions
$\lambda(W)\{M - \E(M\mid W)\}$ for appropriately defined variables $M$
and $W$ and some weight function $\lambda$. It is well known that if
$\E(M\mid W)$ is estimated within a weighted generalized linear model
with canonical link that includes an intercept, then the weighted MLE
estimate solves the score equation
$\sum_i\lambda(W_i)\{M_i - \hat\E(M_i\mid W_i)\}=0$. TMLE uses this
observation to iteratively tilt preliminary estimates of $\q_t$
towards a solution of the efficient influence function estimating
equation. The algorithm is defined as follows:

\begin{enumerate}[label=Step \arabic*., align=left, leftmargin=*]
\item Initialize $\tilde\eta =\hat\eta$ and
  $\tilde\q_{\tau+1,j(i)}(A^\d_{\tau+1,i}, H_{\tau+1,i}) = Y_i$.
\item For $s=1,\ldots,\tau$, compute the weights
  \[\lambda_{s,i} = \prod_{k=1}^s \hat\w_{k,j(i)}(C_{k,i},A_{k,i}, H_{k,i})\]  
\item For $t=\tau,\ldots,1$:
  \begin{itemize}
  \item Compute the pseudo-outcome \corr{$\check Y_{t+1,i} =
    R_{t+1}\tilde\q_{t+1,j(i)}(A^\d_{t+1,i}, H_{t+1,i}) + Z_{t+1}$}
  \item Fit the generalized linear tilting model
    \[\link \tilde\q_t^\epsilon(A_{t,i},H_{t,i}) = \epsilon + \link
      \tilde\q_{t,j(i)}(A_{t,1}, H_{t,i})\] where $\link(\cdot)$ is
    the canonical link. Here, $\link \tilde\q_{t,j(i)}(a_t,h_t)$ is an
    offset variable (i.e., a variable with known parameter value equal
    to one). The parameter $\epsilon$ may be estimated by running a
    generalized linear model of the pseudo-outcome
    $\check Y_{t+1,i}$ with only an
    intercept term, an offset term equal to
    $\link \tilde\q_{t,j(i)}(A_{t,i},H_{t,i})$, and weights
    $\lambda_{t,i}$, using all the data points in the sample such that
    $R_{i,t}=C_{i,t}=1$. An outcome bounded in an interval $[a,b]$ may
    be analyzed with logistic regression (i.e., $\link=\logit$) by
    mapping it to an outcome $(0,1)$ through the transformation
    $(Y - a) / (b - a)(1-2\gamma) + \gamma$ for some small value
    $\gamma>0$. This approach has robustness advantages compared to
    fitting a linear model as it guarantees that the predictions in
    the next step remain within the outcome space \citep[see
    ][]{Gruber2010t}.
  \item Let $\hat\epsilon$ denote the maximum likelihood estimate, and
    update the estimates as
    \begin{align*}
      \link \tilde\q_{t, j(i)}^{\hat\epsilon}(A_{t,i},H_{t,i})&=
                                                                \hat\epsilon + \link\tilde\q_{t, j(i)}(A_{t,i},H_{t,i})\\
      \link \tilde\q_{t, j(i)}^{\hat\epsilon}(A_{t,i}^\d,H_{t,i})&=
                                                                   \hat\epsilon + \link\tilde\q_{t, j(i)}(A_{t,i}^\d,H_{t,i}).
    \end{align*}
    The above procedure with canonical link guarantees the following
    score equation is solved: {\small
      \[\frac{1}{n}\sum_{i=1}^n \left(\prod_{k=1}^t \hat\w_{k,j(i)}(C_{k,i},A_{k,i},
          H_{k,i})\right)\{R_{t+1,i}\tilde \q_{t+1,j(i)}(A^\d_{t+1,i}, H_{t+1,i}) -
        \tilde\q_{t, j(i)}^{\hat\epsilon}(A_{t,i},H_{t,i})\}=0\]
    }
  \item Update $\tilde\q_{t, j(i)}= \tilde\q_{t, j(i)}^{\hat\epsilon}$, $t = t-1$, and iterate. 
  \end{itemize}
\item The TMLE is defined as
  \[\hat\theta=\frac{1}{n}\sum_{i=1}^n\tilde\q_{1,j(i)}(A_{1,i}^\d,L_{1,i}).\]
\end{enumerate}
The iterative procedure and the score equation argument above
guarantee that
\[\corr{\frac{1}{n}\sum_{i=1}^n\sum_{t=1}^\tau\left(\prod_{k=1}^t
    \hat\w_{k,j(i)}(C_{k,i},A_{k,i}, H_{k,i})\right)\{R_{t+1,
    i}\tilde\q_{t+1,j(i)}(A_{t+1,i}^\d, H_{t+1,i}) +Z_{t+1}-
  \tilde\q_{t,j(i)}(A_{t,i}, H_{t,i})\}=0,}\] and thus that
$\Pn \{\D_1(\cdot;\tilde\eta_{j(\cdot)}) - \hat\theta\}= 0$. This fact is
crucial to prove weak convergence and consistency of the TMLE.

Arguing along the lines of the proofs presented in
\cite{diaz2021lmtp}, $\sqrt{n}$-consistency of the above TMLE will
require that the weights $\w_t$ as well as their estimators are all
bounded above, and that the second order term
\begin{equation}
  \sum_{t=1}^\tau\lVert \hat\w_t - \w_t\rVert \lVert \tilde\q_t
  - \q_t \rVert
  \label{eq:1}
\end{equation}
is $o_P(n^{-1/2})$. Note that, for each $t$, we have
\begin{align*}
  \lVert \tilde\q_{t-1}
  - \q_{t-1} \rVert&\leq \lVert \tilde\q_{t-1}
                     - \q_{t-1}^\star \rVert + \lVert \tilde\q_{t-1}^\star
                     - \q_{t-1} \rVert\\
                   &=\lVert \tilde\q_{t-1}
                     - \q_{t-1}^\star \rVert + O_P(\lVert \tilde\q_t
                     - \q_t \rVert).
\end{align*}

Recursive application of the above for $t=\tau-1,\ldots,2$ means
that the term (\ref{eq:1}) is bounded in probability by
\[\sum_{t=1}^\tau\left(\lVert \hat\w_t - \w_t\rVert\sum_{s=t}^{\tau}
    \lVert \hat\q_s - \q_s^\star \rVert\right).\]
Convergence of this term at rate $o_P(n^{-1/2})$ seems to require
convergence of cross-product terms of the type $\lVert \hat\w_t - \w_t
\rVert\lVert \hat\q_s - \q_s^\star \rVert$ for $s\geq t$, which are
not required by the SDR estimator (see Theorem~\ref{theo:asrem}
in the main document).
\section{Distribution of censoring weights in the application}
\begin{figure}[H]
  \centering
  \includegraphics[scale=0.5]{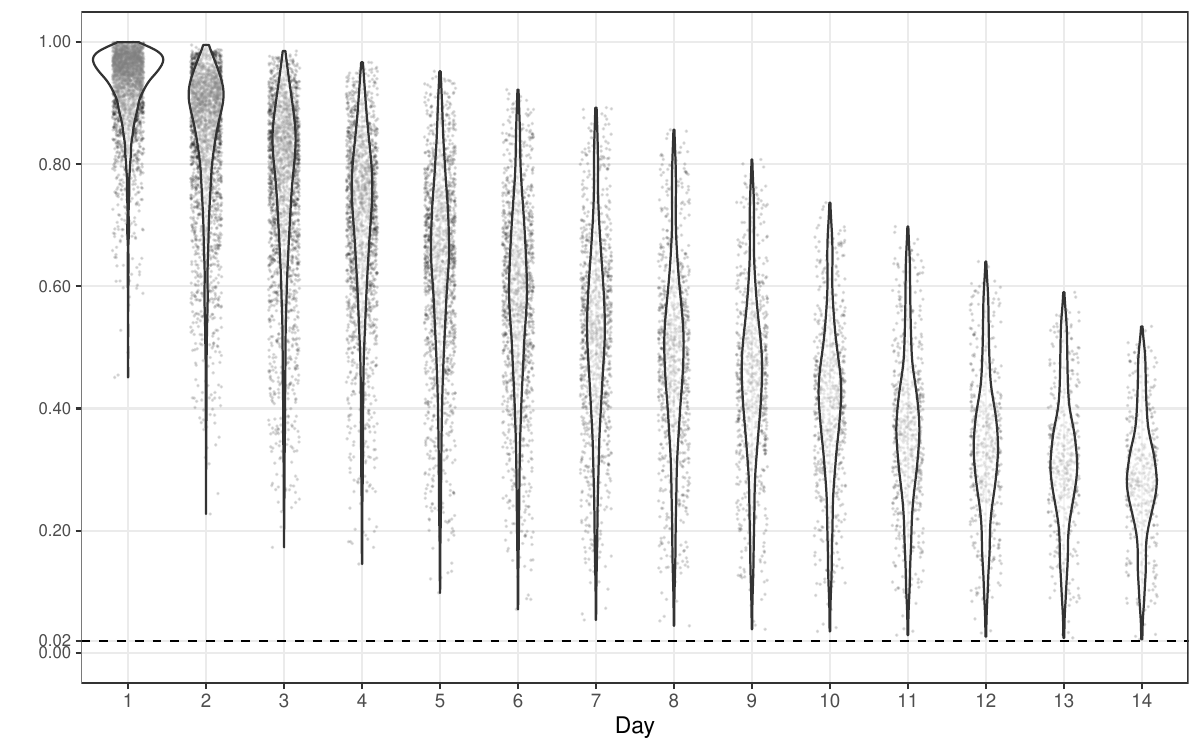}
  \caption{Distribution of the estimated cumulative censoring probabilities
    $\prod_{s=1}^t\hat \g_C(A_{t,i}, H_{t-i})$ for each time $t$.}
  \label{fig:cens}
\end{figure}
\bibliography{refs}
\end{document}